\documentclass[11pt,a4paper,reqno]{amsart}
\usepackage{times}
\usepackage{amsthm,amssymb,amsmath}
\usepackage{pstricks}

\allowdisplaybreaks

\renewcommand{\phi}{\varphi}
\newcommand{\N}{\mathbb{N}}
\newcommand{\R}{\mathbb{R}}
\newcommand{\C}{\mathbb{C}}

\newcommand{\OO}{\mathrm{O}}
\newcommand{\OT}{\mathrm{O}_T}
\newcommand{\Hom}{\mathrm{Hom}}

\newcommand{\Sym}{\mathrm{Sym}}

\newcommand{\Ric}{\mathrm{Ric}}
\newcommand{\tr}{\mathrm{tr}}
\newcommand{\Tr}{\mathrm{Tr}}
\newcommand{\grad}{\mathrm{grad}}
\newcommand{\gradx}{\mathrm{grad}_x}
\renewcommand{\div}{\mathrm{div}}
\newcommand{\Dx}{\Delta_x}
\newcommand{\n}{\nabla}
\newcommand{\nx}{\nabla_x}

\newcommand{\id}{\mathrm{id}}
\renewcommand{\epsilon}{\varepsilon}
\newcommand{\eps}{\varepsilon}
\newcommand{\del}{\partial}
\newcommand{\<}{\langle}
\renewcommand{\>}{\rangle}
\newcommand{\Cinf}{C^\infty}
\newcommand{\Hh}{H_\hbar}
\newcommand{\Hhx}{H_{\hbar,x}}
\newcommand{\MM}{M\bowtie M}
\newcommand{\xy}{\overline{xy}}
\newcommand{\dom}{\mathrm{dom}}

\theoremstyle{plain}
\newtheorem{thm}{Theorem}[section]
\newtheorem{prop}[thm]{Proposition}
\newtheorem{cor}[thm]{Corollary}
\newtheorem{lem}[thm]{Lemma}

\theoremstyle{definition}

\newtheorem{rem}[thm]{Remark}

\newtheorem{exm}[thm]{Example}

\begin{document}

\title{{Asymptotic heat kernel expansion in the semi-classical limit}}

\author{Christian B\"ar, Frank Pf\"affle}

\address{Universit\"at Potsdam\\
Institut f\"ur Mathematik\\
Am Neuen Palais 10\\
14469 Potsdam\\
Germany}

\email{baer@math.uni-potsdam.de,pfaeffle@math.uni-potsdam.de}


\date{\today}

\maketitle

\begin{abstract}
Let $\Hh = \hbar^2L +V$ where $L$ is a self-adjoint Laplace type operator acting on sections of a vector bundle over a compact Riemannian manifold and $V$ is a symmetric endomorphism field.
We derive an asymptotic expansion for the heat kernel of $\Hh$ as $\hbar \searrow 0$.
As a consequence we get an asymptotic expansion for the quantum partition function and we see that it is asymptotic to the classical partition function.
Moreover, we show how to bound the quantum partition function for positive $\hbar$ by the classical partition function.
\end{abstract}

\section{Introduction}

In this paper we study semi-classical approximations for the heat kernel of a general self-adjoint Laplace type operator in a geometric context.
More precisely, let $M$ be an $n$-dimensional compact Riemannian manifold without boundary and let $E\to M$ be a Riemannian or Hermitian vector bundle.
Let $L$ be a self-adjoint Laplace-type operator with smooth coefficients acting on sections of $E$.
Important examples of such operators are the Laplace-Beltrami operator acting on functions, more generally, the Hodge-Laplacian acting on differential forms, and the square of the Dirac operator acting on spinors.
We fix a symmetric endomorphism field $V$ (the potential) which need not be a scalar multiple of the identity.
For any $\hbar>0$ we consider the self-adjoint operator $\Hh := \hbar^2L +V$.
One is now interested in the behavior of $\Hh$ as $\hbar \searrow 0$.

The solution operator $e^{-t\Hh}$ for the heat equation $\frac{\del u}{\del t} + \Hh u = 0$ has a smooth integral kernel $k(x,y,t,\hbar)$ which we briefly call the heat kernel.
Our main result, Theorem~\ref{thm:haupt}, states that there is an asymptotic expansion
$$
k(x,y,t,\hbar)
\quad\stackrel{\hbar \searrow 0}{\sim}\quad
\chi(d(x,y))\cdot q(x,y,t\hbar^2)\cdot \sum_{j=0}^\infty (t\hbar^2)^j \cdot \phi_j(x,y,t) .
$$
Here $d(x,y)$ denotes the Riemannian distance of $x$ and $y$, $\chi$ is a suitable cut-off function, and $q$ is explicitly given by the Euclidean heat kernel, $q(x,y,\tau) = (4\pi \tau)^{-n/2} \cdot \exp\left(-{d(x,y)^2}/{4\tau}\right)$.
The $\phi_j$ are smooth sections which can be determined recursively by solving appropriate transport equations.
Theorem~\ref{thm:haupt} is optimal in the sense that the asymptotic expansion holds for all derivatives with respect to $x$, $y$, and $t$ and is uniform in $x,y\in M$ and $t\in (0,T]$ for any $T>0$.

Now fix $t>0$ and define the quantum partition function $Z_Q(\hbar) := \Tr(e^{-t\Hh})$.
As a direct consequence of Theorem~\ref{thm:haupt} we get Corollary~\ref{cor:Zexpand}
\begin{equation}
Z_Q(\hbar)
\quad\stackrel{\hbar\searrow 0}{\sim}\quad 
(2\sqrt{\pi t}\hbar)^{-n}\cdot \sum_{j=0}^\infty a_j(t)\cdot (t\hbar^{2})^j
\label{eq:ZQto0}
\end{equation}
This result has been proved in \cite[Sec.~2]{TS} by technically rather involved methods from pseudo-differential calculus.
The corresponding classical partition function is given by $Z_C(\hbar) = (2\pi\hbar)^{-n}\int_{T^*M}\tr\left(\exp(-t(|p|^2\cdot\id+V(x)))\right)\,dpdx$.
We directly obtain Corollary~\ref{cor:ZqZc}:
\begin{equation}
\lim_{\hbar\searrow 0}\frac{Z_Q(\hbar)}{Z_C(\hbar)} = 1.
\label{eq:ZQZC}
\end{equation}
This asymptotic statement corresponds to the following bound valid for small positive $\hbar$ (Corollary~\ref{cor:ZAbsch}):
\[
\frac{Z_Q(\hbar)}{Z_C(\hbar)}
\,\,\le\,\, 
c_3 \cdot e^{c_2\cdot t\hbar^2}\cdot \frac{v_{0,n}\left(\sqrt{t\hbar^2}\right)}{v_{K,n}\left(\sqrt{t\hbar^2}\right)} .
\]
Here $v_{K,n}(r)$ denotes the volume of the geodesic ball of radius $r$ in the $n$-dimensional model space of constant curvature $K$.
The constant $c_3$ depends only on the dimension $n$ of $M$, the constant $c_2$ depends on $n$, on a lower bound for the Ricci curvature, and on a lower bound for the potential of $L$.

In the case that the underlying manifold is Euclidean space and the operator $L$ is the classical Laplace operator acting on functions the optimal inequality
$$
\frac{Z_Q(\hbar)}{Z_C(\hbar)}
\,\,\le\,\, 
1
$$
was independently obtained by Golden \cite{G}, Symanzik \cite{S}, and Thompson \cite{T}.
In this classical situation \eqref{eq:ZQZC} can be found in Simon's book \cite[Thm.~10.1]{Sim2}.
In the case of a general compact manifold the asymptotic expansion of $Z_Q(\hbar)$ is contained in \cite[Prop.~2.4]{TS}.

Our methods are rather direct and explicit.
We use standard facts from geometric analysis.
The estimates on the quantum partition function are based on the Golden-Thompson inequality, the Hess-Schrader-Uhlenbrock estimate and an estimate on the heat kernel of the Laplace-Beltrami operator due to Schoen and Yau.
No pseudo-differential calculus or microlocal analysis are needed.

{\em Acknowledgement.}
We would like to thank Markus Klein for very helpful discussions on the topics of this papers and SFB 647 funded by {Deutsche Forschungsgemeinschaft} for financial support.

\section{The heat kernel}\label{sec:kern}

Let $M$ be an $n$-dimensional compact Riemannian manifold without boundary.
In local coordinates the Riemannian metric is denoted by $g_{ij}$, its inverse matrix by $g^{ij}$.

Let $E\to M$ be a Riemannian or Hermitian vector bundle.
We denote the metric on $E$ by $\<\cdot,\cdot\>$.
Let $L$ be a self-adjoint Laplace-type operator with smooth coefficients acting on sections of $E$.
In local coordinates $x=(x^1,\ldots,x^n)$ and with respect to a local trivialization of $E$ we have
$$
L = -g^{ij}(x)\del_i\del_j + b^k(x)\del_k + c(x).
$$
Here the Einstein summation convention is understood, $b^k(x)$ and $c(x)$ are matrices depending smoothly on $x$, and $\del_i = \frac\del{\del x^i}$.

It is well-known \cite[Prop.~2.5]{BGV} that $L$ can globally be written in the form 
$$
L=\n^*\n + W
$$ 
where $\n$ is a metric connection on $E$, $\n^*$ is its formal $L^2$-adjoint, and $W$ is a smooth section of $\Sym(E)$, the bundle of fiberwise symmetric endomorphisms of $E$.

\begin{exm}
Let $E$ be the trivial line bundle, i.~e.\ sections of $E$ are nothing but functions.
Let the connection be the usual derivative, $\n=d$, and let $W=0$.
Then $L=\Delta = d^*d = -\div\circ\grad$, the Laplace-Beltrami operator.
\end{exm}

\begin{exm}
More generally, let $E=\Lambda^pT^*M$ be the bundle of $p$-forms.
Then the Hodge-Laplacian $L=d^*d+dd^*$ is a self-adjoint Laplace-type operator.
Here $\nabla$ is the connection induced on $E$ by the Levi-Civita connection.
For example, for $p=1$, the {\em Bochner formula} says $L=\n^*\n+\mathrm{Ric}$, see
\cite[p.~74, formula (2.51)]{Besse}
or \cite[Chap.~2, Cor.~8.3]{LM}.
\end{exm}

\begin{exm}
Let $M$ carry a spin structure and let $E$ be the spinor bundle.
If $D$ is the Dirac operator, then by the {\em Lichnerowicz-Schr\"odinger formula} $L:=D^2=\n^*\n + \frac14\mathrm{scal}$, see~\cite{Lich} or \cite[Chap.~2, Thm.~8.8]{LM}.
\end{exm}

Now fix another section $V\in \Cinf(M,\Sym(E))$.
For $\hbar>0$ we define the self-adjoint operator
\begin{equation}
\Hh := \hbar^2L + V.
\label{eq:defHh}
\end{equation}
For $t>0$ we use functional calculus to define the operator $e^{-t\Hh}$ as a bounded self-adjoint operator on the Hilbert space of square integrable sections of $E$, $L^2(M,E)$.
For any $u_0\in L^2(M,E)$ we can put $u(x,t):=(e^{-t\Hh}u_0)(x)$ and we get the unique solution to the heat equation
$$
\frac{\del u}{\del t} + \Hh u = 0
$$
subject to the initial condition
$$
u(x,0) = u_0(x).
$$
By elliptic theory $e^{-t\Hh}$ is smoothing and its Schwartz kernel $k(x,y,t,\hbar)$ depends smoothly on all variables $x,y\in M$, $t,\hbar >0$, see \cite[Sec.~2.7]{BGV}.

By $E\boxtimes E^* \to M\times M$ we denote the exterior tensor product bundle of $E$ with its dual bundle $E^*$.
Its fiber over $(x,y)\in M\times M$ is given by $(E\boxtimes E^*)_{(x,y)} = E_x\otimes E^*_y = \Hom(E_y,E_x)$.
Note that for fixed $t$ and $\hbar$ the heat kernel $k(\cdot,\cdot,t,\hbar)$ is a section of $E\boxtimes E^*$.
We define 
$$
q: M\times M \times (0,\infty) \to \R, \hspace{1cm} 
q(x,y,t) := (4\pi t)^{-n/2} \cdot \exp\left(-\frac{d(x,y)^2}{4t}\right),
$$ 
where $d(x,y)$ denotes the Riemannian distance of $x$ and $y$.
For technical reasons we define 
$$
\MM := M\times M \setminus \{(x,y)\,|\,\mbox{$x$ and $y$ are cut-points}\}.
$$
Recall that $x$ and $y$ are cut-points if either there are several geodesics of minimal length joing $x$ and $y$ or there is a Jacobi field along the unique shortest geodesic which vanishes at $x$ and $y$.
For example, on the standard sphere cut points are exactly antipodal points.
One always has that $\MM$ contains the diagonal and is an open dense subset of $M\times M$.
The function $q$ is smooth on $\MM\times (0,\infty)$.

We will often abbreviate $r:=d(x,y)$.
Then $r^2$ is a smooth function on $\MM$ and $r$ itself is smooth on $\MM$ away from the diagonal.

By $\gradx$ we denote the gradient with respect to the variable $x$, similarly for $\Dx$, $\Hhx$, and $\nx$.
Straightforward computation yields
\begin{eqnarray}
\gradx q &=& -\frac{q}{4t}\gradx(r^2) \label{eq:nablaq}\\
\Dx q &=& -q \cdot \left(\frac{r^2}{4t^2} + \frac{\Dx(r^2)}{4t}  \right) \\
\frac{\del q}{\del t} &=& q \cdot \left( \frac{r^2}{4t^2}  - \frac{n}{2t}\right)\\
\left(\frac{\del}{\del t} + \Dx \right)q &=& -q\cdot \frac{2n+\Dx(r^2)}{4t}\label{eq:dtDq}
\end{eqnarray}

\section{The formal heat kernel}\label{sec:formkern}

Now we make the following ansatz for a formal heat kernel of $\Hh$ over $\MM$:
\begin{equation}
\hat k(x,y,t,\hbar) = q(x,y,t\hbar^2)\cdot \sum_{j=0}^\infty (t\hbar^2)^j \cdot \phi_j(x,y,t).
\label{eq:defkhat}
\end{equation}

\begin{lem}\label{lem:formal}
There are unique continuous sections $\phi_j$ over $\MM \times [0,\infty)$, smooth over $\MM \times (0,\infty)$, such that 
\begin{itemize}
 \item[(i)] $\phi_0(y,y,0) = \id_{E_y}$\, for all $y\in M$,
 \item[(ii)] $\left(\frac{\del}{\del t} + \Hhx \right)\hat k = 0$.
\end{itemize}
\end{lem}

Assertion (i) means that for fixed $y$ and $\hbar$ any partial sum of the formal heat kernel $\hat k$ converges to the delta function at $y$ since this is clearly the case for $q$.
Assertion (ii) is to be understood in the sense that the formal series in the definiton \eqref{eq:defkhat} of $\hat k$ is differentiated termwise and then regrouped by powers of $\hbar$.

\begin{proof}[Proof of Lemma~\ref{lem:formal}]
Using \eqref{eq:nablaq},  \eqref{eq:dtDq}, and 
$\n^*\n(f\phi)= \Delta(f)\phi - 2\n_{\grad f}\phi + f\n^*\n\phi$
we compute 
\begin{eqnarray*}
\lefteqn{\left(\frac{\del}{\del t} + \Hhx \right)\hat k}\\
&=&
\sum_{j=0}^\infty\bigg\{  \left(\frac{\del}{\del t} + \hbar^2\Dx \right)(q(x,y,t\hbar^2)) \cdot (t\hbar^2)^j \cdot \phi_j + q\cdot\frac{\del}{\del t}((t\hbar^2)^j \cdot \phi_j)\\
&& \quad\quad
-2\hbar^2\cdot(t\hbar^2)^j\cdot\n_{\gradx q}\phi_j + q\cdot (t\hbar^2)^j\cdot\Hhx \phi_j\bigg \}\\
&=&
q\cdot\sum_{j=0}^\infty (t\hbar^{2})^{j}\bigg\{-\hbar^2\frac{2n+\Dx(r^2)}{4t\hbar^2}\phi_j +\frac{j}{t}\phi_j +\frac{\del}{\del t}\phi_j\\
&& \quad\quad\quad\quad\quad\quad
-2\hbar^2\n_{\frac{-1}{4t\hbar^2}\gradx(r^2)}\phi_j + \hbar^2L_x\phi_j + V\phi_j\bigg \}\\
&=&
q\cdot\sum_{j=0}^\infty t^{j-1}\hbar^{2j}\bigg\{-\frac{2n+\Dx(r^2)}{4}\phi_j +j\phi_j +t\frac{\del}{\del t}\phi_j\\
&& \quad\quad\quad\quad\quad\quad
+\n_{\frac{1}{2}\gradx(r^2)}\phi_j + \hbar^2tL_x\phi_j + tV\phi_j\bigg \}\\
&=&
q\cdot\sum_{j=0}^\infty t^{j-1}\hbar^{2j}\bigg\{\left(j-\frac{2n+\Dx(r^2)}{4}+ tV\right)\phi_j  +t\frac{\del}{\del t}\phi_j\\
&& \quad\quad\quad\quad\quad\quad
+r\n_{\gradx(r)}\phi_j + L_x\phi_{j-1} \bigg \} .
\end{eqnarray*}
Thus assertion (ii) is equivalent to the recursive {\em transport equations}
\begin{equation}
t\frac{\del}{\del t}\phi_j+r\n_{\gradx(r)}\phi_j + \left(j-\frac{2n+\Dx(r^2)}{4}+ tV\right)\phi_j  =- L_x\phi_{j-1}
\label{eq:transport}
\end{equation}
for $j=0,1,\ldots$ where we use the convention $\phi_{-1}\equiv0$.
The function $G:=(2n+\Dx(r^2))/4$ appearing in \eqref{eq:transport} is smooth on $\MM$ and vanishes on the diagonal $\{r=0\}$.

We observe that the transport equation \eqref{eq:transport} is an ordinary differential equation along the integral curves of the vector field $t\frac{\del}{\del t} + r\frac{\del}{\del r}$ in the $r$-$t$-surface, see Fig.~1.
More precisely, if we fix an angle $\theta$ and put $t=\cos(\theta)\cdot s$ and $r=\sin(\theta)\cdot s$, then \eqref{eq:transport} translates into
\begin{equation}
s\frac{\del}{\del s}\phi_j + \left(j-G + \cos(\theta)s V\right)\phi_j  =- L_x\phi_{j-1}
\label{eq:transport2}
\end{equation}
Here we identify the fibers of $E$ by parallel transport along the radial geodesics $\gamma$ emanating from $y$ so that $r\n_{\gradx(r)}$ becomes identified with $r\frac{\del}{\del r}$.

\begin{center}
\begin{pspicture}(0,-2.6770246)(9,2.6770246)
\psset{unit=6mm}
\psbezier[linewidth=0.04](0.02,1.2029753)(1.22,1.0829753)(1.7398921,0.77139497)(2.28,-0.03702468)(2.820108,-0.8454443)(3.26,-1.7370247)(2.76,-2.6570246)
\psbezier[linewidth=0.04](7.44,-1.3970246)(6.3,-0.7570247)(6.02,-0.85702467)(5.22,-1.0370247)(4.42,-1.2170247)(3.54,-1.6170247)(2.76,-2.6370246)
\psbezier[linewidth=0.04](7.44,-1.3970246)(7.72,-0.25702468)(6.122559,1.8689259)(5.56,2.2629752)(4.997441,2.6570246)(0.44,2.2229753)(0.0,1.2029753)
\psdots[dotsize=0.12](3.28,0.08297532)
\psbezier[linewidth=0.04](3.26,0.102975324)(4.22,0.50297534)(4.8,0.76297534)(5.56,0.76297534)(6.32,0.76297534)(6.4,0.6629753)(7.04,0.28297532)
\rput(3.255,-0.3){$y$}
\rput(2.3015625,1.3329753){$M$}
\rput(5.259531,1.2){$\gamma(r)$}

\psline[fillcolor=lightgray,fillstyle=solid,linewidth=0,linecolor=lightgray](10,-3.5)(10,3.5)(13,3.5)(13,-3.5)
\psline[linewidth=0.04cm]{->}(10,-3.5)(10,3.5)
\psline[linewidth=0.04cm]{->}(10,-3.5)(14,-3.5)
\psline[linewidth=0.04cm,linestyle=dotted](13,-3.5)(13,3.5)
\psline[linewidth=0.04cm](10,-3.5)(13,2.5)
\psarc(10,-3.5){2.5}{63}{90}
\rput(9.6,2.9){$t$}
\rput(13.5,-4){$r$}
\rput(10.4,-1.7){$\theta$}
\rput(12.5,2){$s$}
\end{pspicture} 

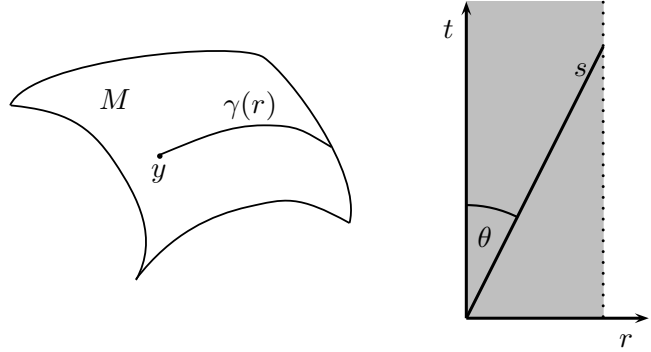
\emph{Fig.~1.} Various parameters and their relation
\end{center}

Since this differential equation is singular at $s=0$ we need to introduce integrating factors.
In order to write them down we solve the linear ODE
\begin{equation}\label{eq:ODE}
\frac{d}{ds}A(s) = A(s)\cdot \cos(\theta)\cdot V(\gamma(\sin(\theta) s)) , \quad\quad A(0)=\id,
\end{equation}
in the space of endomorphisms of $E_y$.
Of course, $A$ depends smoothly on all data such as $s$, $\theta$, $y=\gamma(0)$, and $\dot{\gamma}(0)$.
As long as $A(s)$ is regular the determinant of $A(s)$ satisfies the linear ODE
$$
\frac{d}{ds} \det(A(s)) = \det(A(s))\cdot \cos(\theta)\cdot \tr(V(\gamma(\sin(\theta) s))) ,
$$
hence $\det(A(s))$ cannot vanish anywhere since it would otherwise have to be identically $0$ contradicting the initial condition $\det(A(0))=1$.
Thus $A(s)$ remains invertible for all $s$. 
We now define
$$
R_j(s) := s^j\cdot \exp\left(-\int_0^s\frac{G\,d\sigma}{\sigma}\right)\cdot A(s).
$$
Since $G$ is smooth and vanishes at $s=0$ the integrand $G/\sigma$ is smooth along the segment parametrized by $[0,s]$.
Direct computation shows that \eqref{eq:transport2} is equivalent to 
\begin{equation}
\frac{\del}{\del s}(R_j\phi_j)  = -s^{-1}R_j L_x\phi_{j-1}
\label{eq:transport3}
\end{equation}
For $j=0$ this means that $R_0\phi_0=C_0$ is constant, i.~e.\ $\phi_0=C_0\cdot R_0^{-1} = C_0\cdot \exp\left(\int_0^s\frac{G\,d\sigma}{\sigma}\right)\cdot A(s)^{-1}$.
The initial condition (i) is now equivalent to $C_0=1$.
Thus we have
\begin{equation}
\phi_0=
\exp\left(\int_0^s\frac{G\,d\sigma}{\sigma}\right)\cdot A(s)^{-1} .
\label{eq:phi0}
\end{equation}
For $j\geq 1$ we note that $s^{-1}R_j$ is smooth also at $s=0$ and we get
$R_j\phi_j = -\int_0^s \sigma^{-1}R_jL_x\phi_{j-1}d\sigma + C_j$.
Evaluation at $s=0$ shows $C_j=0$.
We therefore have
\begin{equation*}
\phi_j = -R_j^{-1}\int_0^s \sigma^{-1}R_jL_x\phi_{j-1}d\sigma .
\label{eq:phij}
\end{equation*}
We have established uniqueness of the $\phi_j$.
As to existence we only need to ensure that \eqref{eq:phi0} and \eqref{eq:phij} define {\em smooth} sections over $\MM \times (0,\infty)$.
For $(x,y)\in \MM$ and $r\in [0,1]$ we let $\xy:[0,1]\to M$ be the unique shortest geodesic with $\xy(0)=y$ and $\xy(1)=x$.
In other words, in terms of the Riemannian exponential map $\xy(r) = \exp_y(r\exp_y^{-1}(x))$.
The map $\MM \times [0,1] \to M$, $(x,y,r) \mapsto \xy(r)$, is smooth.
Substituting $\sigma=us$ equation \eqref{eq:phi0} can be rewritten as
\begin{eqnarray}
\phi_0(x,y,t)
&=& 
\exp\left(\int_0^1 u^{-1}G(\xy(u),y)\,du\right) \cdot A(\sqrt{d(x,y)^2+t^2})^{-1}\quad
\label{eq:phi0:2}
\end{eqnarray}
and \eqref{eq:phij} becomes
\begin{eqnarray*}
\lefteqn{\phi_j(x,y,t)}\\
&=& 
-R_j(x,y,t)^{-1}\int_0^1 u^{-1}R_j(\xy(u),y,ut)\cdot(L_x\phi_{j-1})(\xy(u),y,ut)du 
\end{eqnarray*}
where 
\begin{eqnarray*}
\lefteqn{R_j(x,y,t)}\\ 
&=& 
(d(x,y)^2+t^2)^{j/2}\cdot\exp\bigg(-\int_0^1\frac{G(\xy(u),y)\,du}{u} \bigg) \cdot A(\sqrt{d(x,y)^2+t^2}).
\end{eqnarray*}
This shows smoothness of the $\phi_j$ on $\MM\times (0,\infty)$ and continuity on $\MM\times [0,\infty)$.
\end{proof}


\begin{rem}
If $\theta=0$, i.e., if $s=t$ and $x=y$, then \eqref{eq:ODE} becomes a linear ODE with constant coefficients,
$$
\frac{d}{ds}A(s) = A(s)\cdot  V(y) , \quad\quad A(0)=\id,
$$
and can be solved explicitly,
$$
A(s) = \exp(sV(y)).
$$
In particular, since $G(\overline{yy}(u),y)=G(y,y)=0$, \eqref{eq:phi0:2} becomes
\begin{equation}
\phi_0(y,y,t) =  \exp(-tV(y))
\label{eq:phi0x=y}
\end{equation}

\end{rem}

{\bf Construction of the approximate kernel.}
Now we fix $\eta>0$ smaller than the injectivity radius of $M$.
This means that $\{r\leq\eta\} \subset \MM$.
We choose a smooth cutoff function $\chi:\R\to\R$ such that $\chi \equiv 1$ on $(-\infty,\eta/2]$ and $\chi\equiv 0$ on $[\eta,\infty)$.
For $N\in\N$ we set
\begin{equation}
\hat k^{(N)}(x,y,t,\hbar) := \chi(d(x,y))\cdot q(x,y,t\hbar^2)\cdot
\sum_{j=0}^N (t\hbar^2)^j\cdot \phi_j(x,y,t) .
\label{eq:defkhutN}
\end{equation}
Note that $\hat k^{(N)}$ is a smooth section of $E\boxtimes E^*$ over $M\times M\times (0,\infty)\times (0,\infty)$ and not just over $\MM \times (0,\infty)\times (0,\infty)$.

\begin{thm}\label{thm:haupt}
Let $M$ be an $n$-dimensional compact Riemannian manifold without boundary.
Let $E\to M$ be a Riemannian or Hermitian vector bundle and let $\Hh$ be as in \eqref{eq:defHh}.
Let $k(x,y,t,\hbar)$ be the heat kernel of $\Hh$.

Let $T>0$ and $\ell,j\in\N_0$.
For $N>n+\ell+j$ let $\hat k^{(N)}(x,y,t,\hbar)$ be the approximate heat kernel as defined in \eqref{eq:defkhutN}.
Then 
$$
\sup_{t\in(0,T]}\left\|\frac{\del^j}{\del t^j}\left(k-\hat k^{(N)}\right)\right\|_{C^\ell(M\times M)} = \OO(\hbar^{2N+1-2n-2\ell-2j})
\hspace{1cm} (\hbar\searrow 0) .
$$
\end{thm}

\begin{proof}
By construction of $\hat k^{(N)}$ we have (using that $q=\OO((t\hbar^2)^{-n/2})$ and $q=\OO((t\hbar^2)^\infty)$ in the region where $\chi(d(x,y))$ is not locally constant):
\begin{eqnarray*}
r_N 
&:=&
\left(\frac{\del}{\del t} + \Hhx \right)(\hat k^{(N)} - k)\\
&=&
\left(\frac{\del}{\del t} + \Hhx \right)(\hat k^{(N)})\\
&=&
\chi(d(x,y))\cdot\left(\frac{\del}{\del t} + \Hhx \right)\left(q \cdot \sum_{j=0}^N(t\hbar^2)^j\cdot\phi_j\right) +\OO((t\hbar^2)^\infty)\\
&=&
\chi(d(x,y))\cdot q\cdot t^N \hbar^{2N+2} L_x\phi_N+\OO((t\hbar^2)^\infty)\\
&=&
\OO(t^{N-n/2}\hbar^{2N+2-n}) .
\end{eqnarray*}
Similarly, for any $m\in\N$ we get
$$
(\Hhx)^m (r_N) = \OO(t^{N-n/2-2m}\hbar^{2N+2-n-2m}) 
$$
and 
$$
(L_y)^m (r_N) = \OO(t^{N-n/2-2m}\hbar^{2N+2-n-4m}) .
$$
By the choice of the initial condition $\hat k^{(N)} - k$ vanishes at $t=0$.
Thus uniqueness of the solution to the heat equation (Duhamel's principle) implies
\begin{eqnarray*}
(\hat k^{(N)} - k)(x,y,t,\hbar)
&=&
\int_0^t e^{-(t-s)\Hhx}r_N(\cdot,y,s,\hbar)ds .
\end{eqnarray*}
The spectrum of $L$ is bounded from below, hence $\Hh=\hbar^2L+V \geq -c_1$ for all $\hbar\leq 1$ where $c_1$ is a suitable positive constant.
Thus we have for the $L^2$-$L^2$-operator norm
$$
\|e^{-(t-s)\Hhx}\|_{L^2,L^2} \quad\leq\quad e^{c_1t}.
$$
Therefore we get for all $t\in(0,T]$
\begin{eqnarray*}
\|(\hat k^{(N)} - k)(\cdot,y,t,\hbar)\|_{L^2(M)}
&\leq&
\int_0^t \|e^{-(t-s)\Hhx}\|_{L^2,L^2}\cdot\|r_N(\cdot,y,s,\hbar)\|_{L^2(M)}ds\\
&\leq&
\int_0^t e^{c_1t} \cdot c_2 \cdot s^{N-n/2}\hbar^{2N+2-n} ds\\
&\leq&
c_3(T) \cdot \hbar^{2N+2-n} .
\end{eqnarray*}
Furthermore,
\begin{eqnarray*}
\lefteqn{\|L_x(\hat k^{(N)} - k)(\cdot,y,t,\hbar)\|_{L^2(M)}}\\
&=&
\hbar^{-2}\|(\Hhx-V))(\hat k^{(N)} - k)(\cdot,y,t,\hbar)\|_{L^2(M)}\\
&\leq&
\hbar^{-2}\left\|\Hhx\int_0^t e^{-(t-s)\Hhx}r_N(\cdot,y,s,\hbar)ds\right\|_{L^2(M)}\\
&&
+\,\, c_4\hbar^{-2}\|(\hat k^{(N)} - k)(\cdot,y,t,\hbar)\|_{L^2(M)}\\
&\leq&
\hbar^{-2}\left\|\int_0^t e^{-(t-s)\Hhx}\Hhx\, r_N(\cdot,y,s,\hbar)ds\right\|_{L^2(M)}
+ \OT(\hbar^{2N-n})\\
&\leq&
\hbar^{-2}\cdot c_5(T)\cdot\left\|\Hhx\, r_N(\cdot,y,s,\hbar)\right\|_{L^2(M)}
+ \OT(\hbar^{2N-n})\\
&=&
\OT(\hbar^{2N-n-2})
\end{eqnarray*}
and 
\begin{eqnarray*}
\|L_y(\hat k^{(N)} - k)(\cdot,y,t,\hbar)\|_{L^2(M)}
&=&
\left\|\int_0^t e^{-(t-s)\Hhx}L_y\, r_N(\cdot,y,s,\hbar)ds\right\|_{L^2(M)}\\
&=&
\OT(\hbar^{2N-2-n}).
\end{eqnarray*}

Here the lower index $T$ in $\OT(\cdots)$ indicates that the constant bounding the $\OT(\hbar^{2N-n-2})$-term by $\hbar^{2N-n-2}$ depends on $T$.
Integration with respect to $y$ yields
$$
\|(L_x+L_y)(\hat k^{(N)} - k)(\cdot,\cdot,t,\hbar)\|_{L^2(M\times M)}
=\OT(\hbar^{2N-2-n}).
$$
Inductively, we get
$$
\|(L_x+L_y)^m(\hat k^{(N)} - k)(\cdot,\cdot,t,\hbar)\|_{L^2(M\times M)}
= \OT(\hbar^{2N+2-n-4m}).
$$
By the elliptic estimates we have for the Sobolev norms
$$
\|(\hat k^{(N)} - k)(\cdot,\cdot,t,\hbar)\|_{H^{2m}(M\times M)}
= \OT(\hbar^{2N+2-n-4m})
$$
and by the Sobolev embedding theorem
$$
\|(\hat k^{(N)} - k)(\cdot,\cdot,t,\hbar)\|_{C^\ell(M\times M)}
= \OT(\hbar^{2N+1-2n-2\ell}) .
$$
Similarly, we have
$$
\|r_N(\cdot,\cdot,t,\hbar)\|_{C^\ell(M\times M)}
= \OT(\hbar^{2N+1-2n-2\ell}) .
$$
It remains to control the $t$-derivatives.
We compute
\begin{eqnarray*}
\frac{\del}{\del t}(\hat k^{(N)} - k)
&=&
-\Hhx(\hat k^{(N)} - k) + r_N\\
&=&
-\hbar^2 L_x(\hat k^{(N)} - k) - V\cdot(\hat k^{(N)} - k)+ r_N,
\end{eqnarray*}
thus
\begin{eqnarray*}
\lefteqn{\left\|\frac{\del}{\del t}(\hat k^{(N)} - k)\right\|_{C^\ell(M\times M)}}\\
&\leq&
\hbar^2 \cdot c_6\cdot \|\hat k^{(N)} - k\|_{C^{\ell+2}(M\times M)}
+ c_7 \cdot \|\hat k^{(N)} - k\|_{C^\ell(M\times M)}\\
&&
+ \|r_N\|_{C^\ell(M\times M)}\\
&=&
\OT(\hbar^{2N+1-2n-2\ell-2}) .
\end{eqnarray*}
An induction finally proves the theorem.
\end{proof}

\begin{rem}
We were somewhat generous in the application of the Sobolev embedding theorem.
With a little more care we can improve the statement of Theorem~\ref{thm:haupt} as follows:

If in addition to $T,\ell,j$ and $N$ we are given $\eps>0$, then
$$
\sup_{t\in(0,T]}\left\|\frac{\del^j}{\del t^j}\left(k-\hat k^{(N)}\right)\right\|_{C^\ell(M\times M)} = \OO(\hbar^{2N+2-\eps-2n-2\ell-2j})
\hspace{1cm} (\hbar\searrow 0) .
$$
\end{rem}

\begin{rem}
It would be nice to extend Theorem~\ref{thm:haupt} to operator families of the form
$$
\Hh = \hbar^2 L + \hbar D + V
$$
where $D$ is a formally self-adjoint differential operator of first order.
However, this seems not to work.
The resulting transport equations can no longer be solved.

There is a different approximation in \cite{TS} working also for $D\neq0$.
It is less explicit and makes heavy use of asymptotic expansions of total symbols of pseudodifferential operators.
For the trace of the heat operator it seems to give the same result.
We discuss this in the next section.
\end{rem}

\section{The classical and quantum partition functions}\label{sec:trace}

Applying Theorem~\ref{thm:haupt} with $j=\ell=0$ we have in particular
$$
|k(x,y,t,\hbar) - \hat k^{(N)}(x,y,t,\hbar)| \quad\leq\quad C\cdot \hbar^{2N+1-2n}
$$
for all $t\in (0,T]$, $x,y\in M$, and for $\hbar \searrow 0$.
We put $x=y$, take the pointwise trace and integrate over $M$ to obtain

\begin{cor}\label{cor:Zexpand}
Let $M$ be an $n$-dimensional compact Riemannian manifold without boundary.
Let $E\to M$ be a Riemannian or Hermitian vector bundle and let $\Hh$ be as in \eqref{eq:defHh}.
Let $T>0$.
Then we have an asymptotic expansion
$$
\Tr(e^{-t\Hh}) 
\quad\sim\quad 
(2\sqrt{\pi t}\hbar)^{-n}\cdot \sum_{j=0}^\infty a_j(t)t^j\cdot \hbar^{2j}
\hspace{2cm} (\hbar\searrow 0)
$$
uniform in $t\in(0,T]$.
\end{cor}

\begin{proof}
This follows directly from Theorem~\ref{thm:haupt} together with
$q(y,y,t\hbar^2) = (4\pi t\hbar^2)^{-n/2}$ and $a_j(t) = \int_M \tr(\phi_j(y,y,t))\,dy$.
\end{proof}

We fix $t>0$ and we define the {\em quantum partition function}
\begin{equation}
Z_Q(\hbar) := \Tr(e^{-t\Hh})
\label{eq:defZq}
\end{equation}
and the corresponding classical quantity
\begin{equation}
Z_C(\hbar) := (2\pi\hbar)^{-n}\int_{T^*M}\tr\left(\exp(-t(|p|^2\cdot\id+V(x)))\right)\,dpdx .
\label{eq:defZc}
\end{equation}
If $E$ is the trivial line bundle and $V$ is the potential energy this is the partition function in statistical mechanics with $t=1/k_BT$ where $T$ is the temperature and $k_B$ is Boltzmann's constant.
For convenience, we call it the  {\em classical partition function} also in our more general situation.
Then we have

\begin{cor}\label{cor:ZqZc}
Let $M$ be an $n$-dimensional compact Riemannian manifold without boundary.
Let $Z_Q(\hbar)$ and $Z_C(\hbar)$ be defined as in \eqref{eq:defZq} and \eqref{eq:defZc}.
Then
$$
\lim_{\hbar\searrow 0}\frac{Z_Q(\hbar)}{Z_C(\hbar)} = 1.
$$
\end{cor}

\begin{proof}
Using $\int_{\R^n}\exp(-t|p|^2)\,dp = (\pi/t)^{n/2}$  and \eqref{eq:phi0x=y} we get
\begin{eqnarray}
Z_C(\hbar)
&=&
(2\pi\hbar)^{-n}  \cdot\int_{T^*M}\tr\left(\exp(-t|p|^2)\exp(-t\cdot V(x))\right)\,dpdx\nonumber\\
&=&
(2\pi\hbar)^{-n} \cdot(\pi/t)^{n/2} \cdot\int_M\tr(\exp(-t\cdot V(x)))\,dx\nonumber\\
&=&
(2\sqrt{\pi t}\hbar)^{-n}\cdot\int_M \tr(\phi_0(x,x,t))\,dx\nonumber\\
&=&
(2\sqrt{\pi t}\hbar)^{-n}\cdot a_0(t) .\nonumber
\end{eqnarray}
The assertion follows from Corollary~\ref{cor:Zexpand}.
\end{proof}

In the remainder of this section we will contrast this asymptotic comparison of $Z_Q(\hbar)$ and $Z_C(\hbar)$ with an inequality of the two partition functions which works for positive $\hbar$.
For this we need the following version of the Golden-Thompson inequality (see \cite{ENT}, \cite{Len} or \cite{Sim}).

\begin{lem}\label{lem:Golden-Thompson}
Let $B$ and $C$ be self-adjoint operators on a Hilbert space $H$, both bounded from below and such that $B+C$ is essentially self-adjoint on the intersection $\dom(B)\cap\dom(C)$ of the domains of $B$ and $C$. Then 
\[ 
\Tr\left( \exp(-(B+C))\right)\le \Tr\left( \exp(-B)\exp(-C)\right).
\]
\end{lem}
We will also need the following elementary assertion.

\begin{lem}\label{lem:traceestimate}
Let $A_1$ and $A_2$ be complex $N\times N$-matrices.
Let $A_2$ be Hermitian and nonnegative.
Then 
\begin{equation}
|\tr\left(A_1\, A_2\right)|\le |A_1|\cdot \tr\left(A_2\right)
\label{eq:trNtr}
\end{equation}
where $|A_1|$ denotes the operator norm of $A_1$.
\end{lem}

\begin{proof}
Since $A_2$ is Hermitian and nonnegative there is a nonnegative Hermitian matrix $B$ such that $A_2=B^2$.
For the standard basis $e_1,\ldots,e_N$ of $\C^N$ and the standard Hermitian scalar product $(\cdot,\cdot)$ we get
\begin{eqnarray*}
|\tr\left(A_1\, A_2\right)|
&=&
|\tr\left(BA_1B\right)|
\,\,=\,\,
\bigg|\sum_{j=1}^N (BA_1Be_j,e_j)\bigg|
\,\,=\,\,
\bigg|\sum_{j=1}^N (A_1Be_j,Be_j)\bigg|\\
&\leq&
\sum_{j=1}^N |A_1|\cdot |Be_j|^2
\,\,=\,\,
|A_1|\cdot\sum_{j=1}^N (B^2e_j,e_j)
\,\,=\,\,
|A_1|\cdot \tr\left(A_2\right) .
\end{eqnarray*}
\end{proof}

\begin{prop}\label{prop_ZQupperbound}
Let $M$ be an $n$-dimensional compact Riemannian manifold without boundary.
For any $x\in M$ and $r>0$ let $\omega_x(r)$ denote the volume of the geodesic ball about $x$ with radius $r$. 
Then there are constants $c_1>0$ and $c_2\in\R$ such that for any $t>0$ and any $\hbar >0$  one has
\begin{equation}\label{ineq:ZQvolume}
Z_Q(\hbar)\le c_1\cdot e^{c_2\cdot t\hbar^2}\int_M\;\frac{\tr\left(e^{-tV(x)}\right)}{\omega_x(\sqrt{t\hbar^2})} \, dx.
\end{equation}
\end{prop}
\begin{proof}
We apply Lemma~\ref{lem:Golden-Thompson} with $B=t\hbar^2\nabla^*\nabla$ and $C=t(V+\hbar^2 W)$ and use the notations $k_{t}^{\nabla^*\nabla}(x,y)$ and $k_{t}^{\Delta}(x,y)$ for the kernels of the heat operators $e^{-t\nabla^*\nabla}$ and $e^{-t\Delta}$ respectively, where $\Delta$ is the Laplace-Beltrami operator acting on functions:
\begin{eqnarray*}
Z_Q(\hbar)
&\le& 
\Tr\left(e^{-t\hbar^2\nabla^*\nabla}e^{-t(V+\hbar^2 W)} \right)\\
&=& 
\int_M \tr\left(k_{t\hbar^2}^{\nabla^*\nabla}(x,x)e^{-t(V(x)+\hbar^2 W(x))} \right)dx\\
&\stackrel{\eqref{eq:trNtr}}{\le}& 
\int_M \left|k_{t\hbar^2}^{\nabla^*\nabla}(x,x)\right|\cdot\tr\left(e^{-t(V(x)+\hbar^2 W(x))} \right)dx\\
&\le& 
\int_M k_{t\hbar^2}^{\Delta}(x,x)\;\tr\left(e^{-t(V(x)+\hbar^2 W(x))} \right)dx
\end{eqnarray*}
where for the last inequality we have the Hess-Schrader-Uhlenbrock inequality, i.e.\ $\left|k_{t}^{\nabla^*\nabla}(x,y)\right|\le k_{t}^{\Delta}(x,y)$ for any $x,y\in M$ and $t>0$ (see \cite{HSU}).
Applying Lemma~\ref{lem:Golden-Thompson} once more (with $B=tV(x)$ and $C=t\hbar^2 W(x)$), using Lemma~\ref{lem:traceestimate} and choosing a $w_0\in\R$ with $W\ge w_0$ leads to 
\begin{eqnarray}
Z_Q(\hbar)
&\le& \int_M k_{t\hbar^2}^{\Delta}(x,x)\;\tr\left(e^{-tV(x)}e^{-t\hbar^2 W(x)} \right)dx\nonumber \\
&\le& \int_M k_{t\hbar^2}^{\Delta}(x,x)\;e^{-t\hbar^2 w_0} \;\tr\left(e^{-tV(x)}\right)dx\label{eq:ZQ_kernelestimate} .
\end{eqnarray}
From~\cite[Thm.~4.6]{SY} we get a pointwise estimate for the heat kernel of the Laplace-Beltrami operator:
There are constants $c_1, \widetilde{c}>0$ such that for any $\tau >0$ and any $x\in M$ one has
\begin{equation}\label{eq:estimateschoenyau}
k_\tau^\Delta(x,x)\le c_1\cdot \frac{1}{\omega_x(\sqrt{\tau})}\cdot e^{\widetilde{c}\cdot \tau}.
\end{equation}
Inserting this into (\ref{eq:ZQ_kernelestimate}) and taking $c_2=\widetilde{c}-w_0$ yields the claim.
\end{proof}
\begin{rem}
The constants $c_1,c_2$ can be given explicitly:
Let $\alpha>1$, $\delta>0$ and $n=\dim(M)$, let $\kappa>0$ with $\Ric\ge -\kappa$ and let $w_0\in\R$ with $W\ge w_0$.
In (\ref{eq:estimateschoenyau}) one can take $c_1=(1+\delta)^{n\alpha}\exp\left(\frac{1+\alpha}{\delta}\right)$ and $\widetilde{c}=\frac{\alpha n}{\alpha-1}\cdot \kappa\cdot \delta$ (compare~\cite[Thm.~4.6]{SY}) and then $c_2=\widetilde{c}-w_0$.
\end{rem}

Given $K\in\R$ and $r\in(0,\frac{\pi}{\sqrt{K}})$ (where we use the convention $\frac{\pi}{\sqrt{K}}=\infty$ for $K\leq0$) let $v_{K,n}(r)$ denote the volume of a geodesic ball of radius $r$ in the $n$-dimensional model space of constant curvature $K$.
Recall that this model space is hyperbolic space, Euclidean space, or the sphere with their appropriately scaled standard metrics.

\begin{cor}\label{cor:ZAbsch}
Let $M$ be an $n$-dimensional compact Riemannian manifold without boundary, and let $K\in\R$ be an upper bound for the secional curvature of $M$ and let $\iota$ denote the injectivity radius of $M$.
Then there are constants $c_2,c_3\in\R$ such that for any $t,\hbar>0$ with $t\hbar^2<\frac{\pi^2}{K}$ and $t\hbar^2<\iota^2$ we have
\[
\frac{Z_Q(\hbar)}{Z_C(\hbar)}
\,\,\le\,\, 
c_3 \cdot e^{c_2\cdot t\hbar^2}\cdot \frac{v_{0,n}\left(\sqrt{t\hbar^2}\right)}{v_{K,n}\left(\sqrt{t\hbar^2}\right)} .
\]
\end{cor}

\begin{proof}
For $r<\min\{\iota,\frac{\pi}{\sqrt{K}}\}$ the Bishop-G\"unther Theorem \cite[Thm.~3.7]{Chavel} states $\omega_x(r)\ge v_{K,n}(r)$ for all $x\in M$.
Inserting this into (\ref{ineq:ZQvolume}) yields
\begin{eqnarray*}
Z_Q(\hbar)
&\le& 
c_1\cdot e^{c_2\cdot t\hbar^2}\cdot \frac{1}{v_{K,n}\left(\sqrt{t\hbar^2}\right)}\cdot \int_M\;\tr\left(e^{-tV(x)}\right)dx\\
&=&
c_1\cdot e^{c_2\cdot t\hbar^2}\cdot \frac{1}{v_{K,n}\left(\sqrt{t\hbar^2}\right)}\cdot (2\sqrt{\pi t}\hbar)^n \cdot Z_C(\hbar)
\end{eqnarray*}
for any $t,\hbar>0$ with $t\hbar^2<\frac{\pi^2}{K}$ and $t\hbar^2<\iota^2$.
Putting $c_3 :=  c_1 \cdot \frac{(2\sqrt{\pi})^n}{v_{0,n}(1)}$ concludes the proof.
\end{proof}

\begin{rem}
Even optimal choices of $\alpha$ and $\delta$ yield a constant $c_3$ which is much larger than $1$.
Therefore Corollary~\ref{cor:ZAbsch} does not even imply half of Corollary~\ref{cor:ZqZc}, namely 
$$
\lim_{\hbar\searrow 0}\frac{Z_Q(\hbar)}{Z_C(\hbar)} \,\,\leq\,\, 1.
$$
We do not know whether Corollary~\ref{cor:ZAbsch} holds with $c_3
=1$.
\end{rem}

\end{document}